\documentclass{cccg-archiv}
\usepackage{amsmath}
\usepackage{amsfonts}
\usepackage{amssymb}
\usepackage{enumerate}
\usepackage[T1]{fontenc}
\usepackage[utf8]{inputenc}
\usepackage{color}
\usepackage{graphicx}
\usepackage{gensymb}

\renewenvironment{thebibliography}[1]{%
   \begin{oldthebibliography}{#1}%
     \setlength{\itemsep}{-.5ex}%
}%
{%
   \end{oldthebibliography}%
}

\newcommand{\bb}[1]{\ensuremath\mathbf{#1}}
\newcommand{\sw}{{{{\tt L}}}}
\newcommand{\nw}{\raisebox{\depth}{\scalebox{1}[-1]{{{\sw}}}}}

\newcommand{\nec}{\reflectbox{\nw}}
\newcommand{\se}{\reflectbox{\sw}}
\newcommand{\updownF}{\raisebox{\depth}{\scalebox{1}[-1]{{{\tt F}}}}}

\newtheorem{definition}{Definition}

\title{$1$-String $B_1$-VPG Representations of Planar Partial $3$-Trees \\ and Some Subclasses}

\author{Therese Biedl\thanks{David R. Cheriton School of Computer Science, University of Waterloo, {\tt biedl@uwaterloo.ca}. Research supported by NSERC.}
        \and
        Martin Derka\thanks{David R. Cheriton School of Computer Science, University of Waterloo, {\tt mderka@uwaterloo.ca}. Research supported by an NSERC Vanier CGS.}}

\index{Biedl, Therese}
\index{Derka, Martin}

\begin{document}
\thispagestyle{empty}
\maketitle

\abstract{
Planar partial $3$-trees are subgraphs of those planar graphs obtained by repeatedly inserting a vertex of degree $3$ into a face. In this paper, we show that planar partial $3$-trees have $1$-string $B_1$-VPG representations, i.e., representations where every vertex is represented by an orthogonal curve with at most one bend, every two curves intersect at most once, and intersections of curves correspond to edges in the graph. We also that some subclasses of planar partial 3-trees have \{\sw{}\}-representations, i.e., a $B_1$-VPG representation where every curve has the shape of an \sw{}.
}

\section{Introduction}

A~\emph{string representation} is a representation of a graph where every vertex $v$ is assigned a curve $\bb{v}$. Vertices $u,v$ are connected by an edge if and only if curves $\bb{u},\bb{v}$ intersect. A \emph{$1$-string representation} is a string representation where every two curves intersect at most once. 

String representations of planar graphs were first investigated by  Ehrlich, Even and Tarjan in 1976~\cite{cit:tarjan}. They showed that every planar graph has a $1$-string representation using ``general'' curves. In 1984, Scheinerman conjectured~\cite{cit:scheinerman} that every planar graph has a 
$1$-string representation, and furthermore curves are line segments. Chalopin, Gon\c{c}alves and Ochem~\cite{cit:chalopin-gonclaves-ochem,cit:chalopin-string} proved that every planar graph has a $1$-string representation in 2007. Scheinerman's conjecture itself remained open until 2009 when it was proved true by Chalopin and Gon\c{c}alves~\cite{cit:chalopin-seg}. 

Our paper investigates string representations that use \emph{orthogonal curves}, i.e., curves consisting of vertical and 
horizontal segments. If every curve has at most $k$ bends, these are called {\em $B_k$-VPG representations}. The hierarchy of $B_k$-VPG representations was introduced by Asinowski et al.~\cite{cit:asinowski, cit:asinowski-vpg}. VPG is an acronym for Vertex-Path-Grid since vertices are represented by paths in a rectangular grid. 

It is easy to see that all planar graphs are VPG-graphs (e.g. by generalizing the construction of Ehrlich, Even and Tarjan). For bipartite planar graphs, curves can even be required to have no
bends~\cite{cit:arroyo, cit:pach}. For arbitrary planar graphs, bends in orthogonal curves are required. Chaplick and Ueckerdt showed that 2 bends per curve always suffice~\cite{cit:chaplick}.  In a recent paper we 
strengthened this to give a $B_2$-VPG representation that is also a 1-string representation
\cite{cit:socg}.

$B_k$-VPG representations were further studied by Chaplick, Jel\'{i}nek, Kratochov\'{i}l and Vysko\v{c}il~\cite{cit:kratochvil} who showed that recognizing $B_k$-VPG graphs is NP-complete even when the input graph is given by a $B_{k+1}$-VPG representation, and that for every $k$, the class of $B_{k+1}$-VPG graphs is strictly larger than $B_k$-VPG. 

\paragraph*{Our Contribution}
Felsner et al.~\cite{cit:mfcs} showed that every planar 3-tree has a $B_1$-VPG representation.  Moreover, every vertex-curve has the shape of an {\tt L} (we call this an {\em \{{\tt L}\}-representation}). This implies that any two vertex-curves intersect at most once, so this is a 1-string $B_1$-VPG representation.  In this paper, we extend the result to more graphs, and in particular, show:

\begin{theorem}
\label{thm:result}
Every planar partial $3$-tree $G$ has a $1$-string $B_1$-VPG representation.
\end{theorem}


There are 4 possible shapes of orthogonal curves with one bend. Depending of where the bend is situated, we call them
\sw{}, \nw{}, \nec{} and \se{} 
respectively. Note that a horizontal or vertical curve without bends can be turned into any of the shapes by adding one bend. 

The construction of our proof of Theorem~\ref{thm:result} uses all 4 possible shapes 
\sw{}, \nw{}, \nec{}, \se{}.  However, for some subclasses of planar partial $3$-trees, 
we can show that fewer shapes suffice.  We use the notation $\{\sw{},\nec{}\}$-representation for a
$B_1$-VPG representation where all curves are either \sw{} or \nec{}, and similarly for other subsets of
shapes.  We can show the following:

\begin{theorem}
\label{thm:IO}
Any IO-graph has an $\{{\tt L}\}$-representation. 
\end{theorem}

\begin{theorem}
\label{thm:Halin}
Any Halin-graph has an $\{\sw{},\nec{}\}$-representation, and only one
vertex uses a \nec{}-shape. 
\end{theorem}

We give the definitions of these graph classes and the proof of these theorems
in the next three sections, and end with open problems in Section~\ref{sec:conclusion}.

\section{Planar partial 3-trees}
\label{sec:def}
\label{sec:proof}

A {\em planar graph} is a graph that can be drawn without edge crossings.
If one such drawing $\Gamma$ is fixed, then a {\em face} is a maximal connected
region of $\Bbb{R}^2-\Gamma$. The {\em outer face} corresponds to the unbounded region;
the {\em interior  faces} are all other faces.  A vertex is called {\em exterior} if it
is on the outer face and {\em interior} otherwise.

A {\em $3$-tree} is a graph that is either a triangle or has a vertex  order $v_1,\dots,v_n$ such that for $i\geq 4$, vertex $v_i$ is adjacent to exactly three predecessors and they form a triangle.  
A \emph{partial $3$-tree} is a subgraph of a $3$-tree.   

Our proof of Theorem~\ref{thm:result} employs the method of ``private regions'' used previously for various string representation constructions~\cite{cit:socg, cit:chalopin-string, cit:mfcs}. We define the following:

\begin{definition}[F-shape and rectangular shape] An \emph{F-shaped area} is a region bounded by a $10$-sided polygon with CW or CCW sequence of interior angles $90\degree$, $270\degree$, $90\degree$, $90\degree$, $270\degree$, $270\degree$, $90\degree$, $90\degree$, $90\degree$ and $90\degree$.  A~\emph{rectangle-shaped area} is a region bounded by an axis-aligned rectangle.
\end{definition}

\begin{definition}[Private region]
Given a $1$-string representation,
a \emph{private region} of vertices $\{a,b,c\}$
is an F-shaped or rectangle-shaped 
area that intersects (up to permutation of names) curves
$\bb{a}, \bb{b}, \bb{c}$ in the way depicted in 
Figure~\ref{fig:private-region}(a), and that intersects no other curves and private regions.
\end{definition}

\begin{figure}[ht]
\centering
\includegraphics[width=\columnwidth]{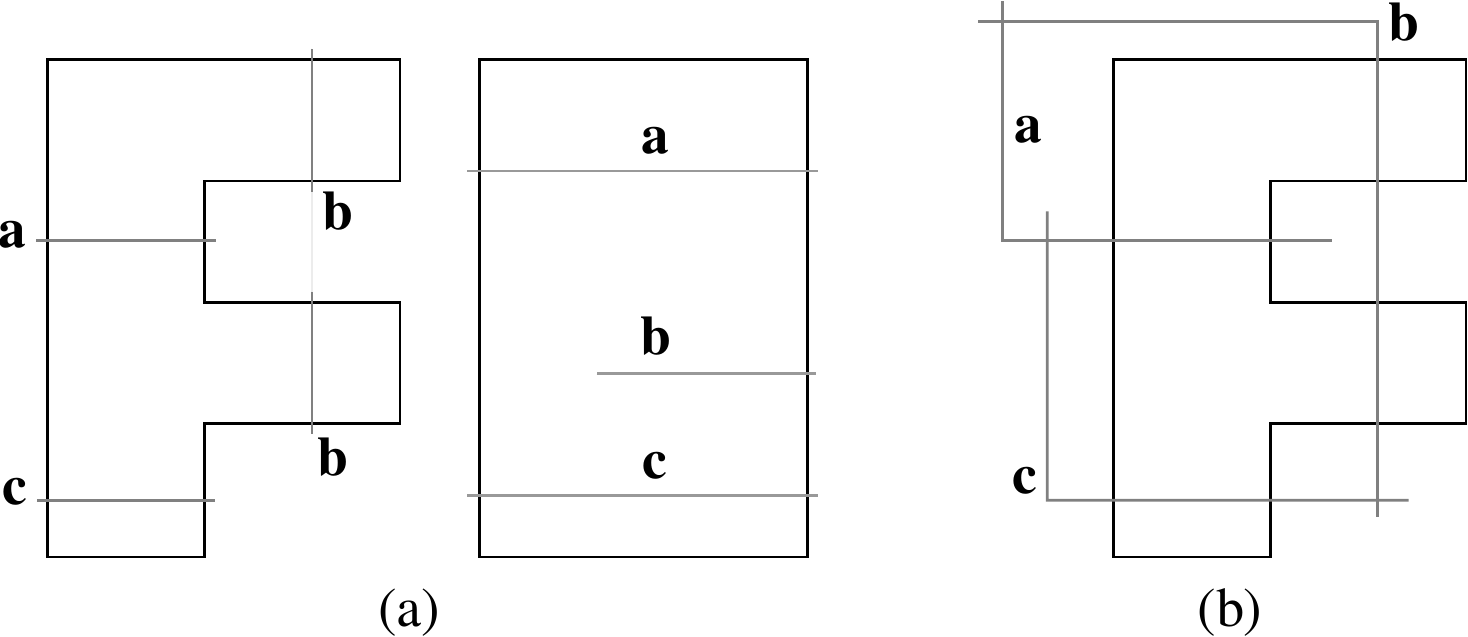}
\caption{(a) An F-shaped 
(left) 
and rectangle-shaped 
(right) 
private region of 
$\{a,b,c\}$.
(b) The base case. Intersections among $\{\bb{a},\bb{b},\bb{c}\}$ can be omitted as needed.}
\label{fig:private-region}
\label{fig:basecase}
\end{figure}

Now we are ready to prove Theorem~\ref{thm:result}. 
Let $G$ be a planar partial $3$-tree.   By definition, there exists
a 3-tree $H$ for which $G$ is a subgraph.  One can show~\cite{cit:biedl}
that we may assume $H$ to be planar.  Let $v_1,\dots,v_n$ be a vertex
order of $H$ such that for $i\geq 4$ vertex $v_i$ is adjacent to 3 predecessors
that form a triangle.  
In particular, $v_4$ is incident to a triangle formed by $\{v_1,v_2,v_3\}$.
One can show (see e.g.~\cite{cit:biedl}) that 
the vertex order can be chosen in such a way that $\{v_1,v_2,v_3\}$ is
the outer face of $H$ in some planar drawing.

For $i\geq 3$, let $G_i$ and $H_i$ be the subgraphs of $G$ (respectively
$H$) induced by vertices $v_1,\dots,v_i$.
We prove Theorem~\ref{thm:result} by showing the following by induction on $i$:
\begin{quotation}
\noindent
$G_i$ has a $1$-string $B_1$-VPG representation with a private 
region for every interior face of $H_i$. 
\end{quotation}

In the base case, $i = 3$ and $G \subseteq K_3 \simeq H$. Construct a representation $R$ and find a private region for the unique interior face of $H$
as depicted in Figure~\ref{fig:basecase}(b). 

Now consider $i \geq 4$.  
By induction, construct a representation $R_0$ of $G_{i-1}$ that contains a private region for 
every interior face of $H_{i-1}$. 

Let $\{a,b,c\}$ be the predecessors of $v_i$ in $H$. Recall that they form a triangle.
Since $H$ is planar, this triangle must form a face in $H_{i-1}$. 
Since $\{v_1,v_2,v_3\}$ is the outer face of $H$ (and hence also of $H_{i-1}$), 
the face into which $v_i$ is added must be an interior face, so there exists an interior face 
$\{a,b,c\}$ in $H_{i-1}$.  Let $P_0$ be the private region that exists for $\{a,b,c\}$ in $R_0$; 
it can have the shape of an F or a rectangle. 

Observe  that in $G$, vertex $v_i$ may 
be adjacent to any possible subset of $\{a,b,c\}$.
This gives 16 cases (two possible shapes, up to rotation and reflection, and 8 possible adjacencies).

In each case, the goal is to place a curve $\bb{v_i}$ inside $P_0$ such that it intersects exactly
the curves of the neighbours of $v_i$ in $\{a,b,c\}$ and no other curve.
Furthermore, having placed $\bb{v_i}$ into $P_0$, we need to find a private region for the
three new interior faces in $H_i$,  that is, the three faces formed by $v_i$ and two of 
$\{a, b, c\}$.   

\smallskip
\noindent\textbf{Case 1: $P_0$ has the shape of an F.}
After possible rotation / flip of $R_0$ and renaming of $\{a,b,c\}$ we may assume that $P_0$
appears as in Figure~\ref{fig:private-region}(a).  If $(v_i,a)$ is an edge, then place a bend
for curve $\bb{v_i}$ in the region above $\bb{a}$.  Let the vertical segment of $\bb{v_i}$
intersect $\bb{a}$ and (optionally) $\bb{c}$.  Let the horizontal segment of $\bb{v_i}$
intersect (optionally) the top occurrence of $\bb{b}$.
If $(v_i,a)$ is not an edge but $(v_i,c)$ is an edge, then place a bend for $\bb{v_i}$ in the region below $\bb{a}$, 
let the vertical segment of $\bb{v_i}$ intersect $\bb{c}$ and the horizontal segment of $\bb{v_i}$ intersect
(optionally) $\bb{b}$.
Finally, if neither $(v_i,a)$ nor $(v_i,c)$ is an edge, then $\bb{v_i}$ is a horizontal segment in the
region below $\bb{a}$ and above $\bb{c}$ that (optionally) intersects $\bb{b}$.

In all sub-cases, $\bb{v_i}$ remains inside $P_0$, so it cannot intersect any other curve of $R_0$. 
Private regions for the newly created faces can be found as shown in Figure~\ref{fig:efs}.

\begin{figure}[ht]
\centering
\includegraphics[width=0.3\columnwidth,page=2]{f-abbrv}
\includegraphics[width=0.3\columnwidth,page=1]{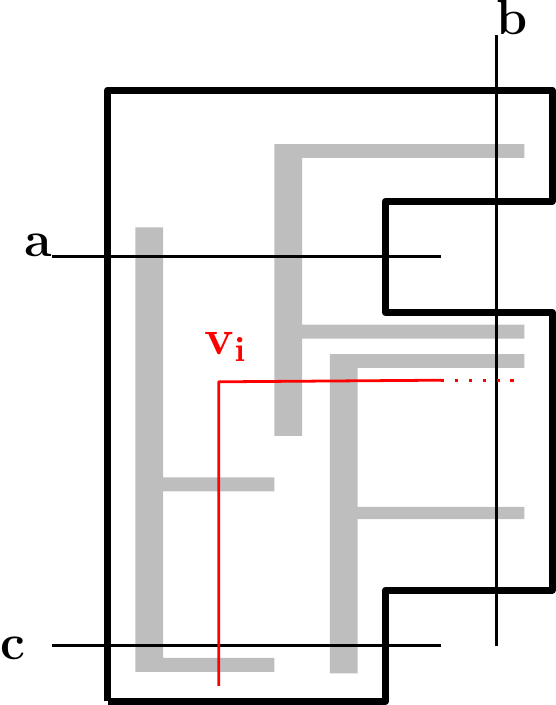}
\includegraphics[width=0.3\columnwidth,page=3]{f-abbrv}
\caption{Inserting curve $\bb{v_i}$ into an F-shaped private region.  (Left) $(v_i,a)$ is an edge.
(Middle) $(v_i,a)\not\in E$, but $(v_i,c)\in E$.  (Right) $(v_i,a),(v_i,c)\not\in E$.}
\label{fig:efs}
\end{figure}

\smallskip
\noindent\textbf{Case 2: $P_0$ has the shape of a rectangle.}
After possible rotation / flip of $R_0$ and renaming of $\{a,b,c\}$ we may assume that $P_0$
appears as in Figure~\ref{fig:private-region}(a).  If $(v,a)$ is an edge, then $\bb{v}$
is a vertical segment that intersects $\bb{a}$ and (optionally) $\bb{b}$ and (optionally)
$\bb{c}$.  If $(v,c)$ is an edge, then symmetrically $\bb{v}$ is a vertical segment that
intersects $\bb{c}$ and (optionally) $\bb{b}$ and $\bb{a}$.  Finally if neither $(v,a)$ nor
$(v,c)$ is an edge, then let $\bb{v}$ be a horizontal segment between $\bb{a}$ and $\bb{c}$
with (optionally) a vertical segment attached to create an intersection with $\bb{b}$.

In all cases, $\bb{v}$ remains inside $P_0$, so it cannot intersect any other curve of $R_0$. 
Private regions for the newly created faces can be found as shown in Figure~\ref{fig:rectangles}.

\begin{figure}[ht]
\centering
\includegraphics[width=0.3\columnwidth,page=2]{rectangles-abbrv}
\includegraphics[width=0.3\columnwidth,page=1]{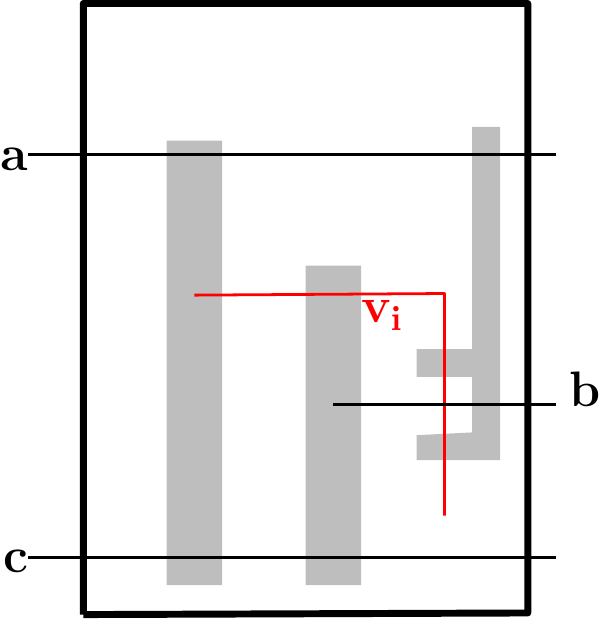}
\includegraphics[width=0.3\columnwidth,page=3]{rectangles-abbrv}
\caption{Inserting curve $\bb{v_i}$ into a rectangle-shaped private region.
(Left) $(v_i,a)$ is an edge.
(Middle) $(v_i,a),(v_i,c)\not\in E$, but $(v_i,b)\in E$.
(Right) $(v_i,a),(v_i,b),(v_i,c)\not\in E$.
}
\label{fig:rectangles}
\end{figure}

Theorem~\ref{thm:result} now holds by induction. \hfill $\square$

\bigskip

We note here that in our proof-approach, both types of private regions and all four
shapes with one bend are required in some cases.

\section{IO-Graphs}
\label{sec:other-classes}
\label{sec:IO}


An \emph{IO-graph}~\cite{cit:iographs} is a 2-connected planar graph with a planar embedding such that the interior vertices form a (possibly empty) independent set. 
One can easily show \cite{cit:iographs} that every IO-graph is a planar partial $3$-tree.

We now prove Theorem~\ref{thm:IO} by constructing an $\{\sw{}\}$-representation of an IO-graph $G$.
Let $O$ be the set of exterior vertices; by definition these induce an {\em outerplanar graph}, i.e., a
graph that can be embedded so that all vertices are on the outer face.   Moreover, since $G$ is
2-connected, the outer face is a simple cycle, and hence the outerplanar graph $G[O]$ is also
2-connected.  We first construct an $\{\sw{}\}$-representation of $G[O]$, and then insert the
interior vertices.  To do so, we again use private regions, but we modify their definition slightly
in three ways:  (1) Interior vertices may have arbitrarily high degree, and so the private
regions must be allowed to cross arbitrarily many curves. (2)  Interior vertices may only be
adjacent to exterior vertices.  It therefore suffices for the private region of
f face $f$ to intersect only those curves that belong to exterior vertices on $f$.
It is exactly this latter observation that allows
us to find private regions more easily, therefore use fewer shapes for them, and 
therefore use fewer shapes for the curves.  We can therefore also add:  (3) The private
region must be an F-shape, and it must be in the rotation \updownF.  The formal definition
is given below:

\begin{definition}[IO-private region]
Given a 1-string representation of an IO-graph, an \emph{IO-private region} of a face $f$ 
is an $F$-shaped area $P$, in the rotation \updownF, which intersects
curves $\bb{x_1}, \bb{x_2}, \ldots, \bb{x_d}$ as shown in Figure~\ref{fig:long-private}. 
Here, $\{x_1,\dots,x_d\}$ is a subset of the vertices of $f$ enumerated in CCW order,
and includes all exterior vertices that belong to $f$ (it may or may not include other vertices).
Lastly, $P$ intersects no other curves and no other private regions. 
\end{definition}

\begin{figure}
\centering
\includegraphics[width=.4\columnwidth,page=1]{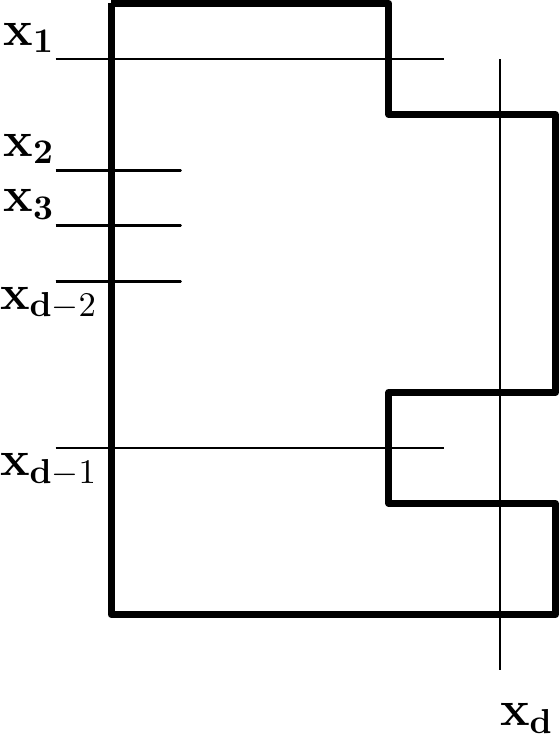}
\caption{An IO-private region.
We require that the supporting line of $\bb{x_i}$ (for $i=2,\dots,k-2$) 
intersects the upper segment of $\bb{x_d}$.} 
\label{fig:long-private}
\end{figure}

\begin{lemma}
\label{thm:skeleton-private}
\label{lem:skeleton-private}
Any outer planar graph has an \{{\tt\em L}\}-representation with an IO-private region 
for every interior face.  
\end{lemma}
\begin{proof}
We may assume that the outerplanar graph is 2-connected, otherwise we can add vertices to make
it so and delete their curves later.  
Enumerate the vertices on the outer face as $v_1,\dots,v_k$ in CCW order.
For every vertex $v_i$ on the outer-face, let $\bb{v_i}$ be an {\tt L} with bend at $(i,-i)$.   
The vertical segment reaches of $\bb{v_i}$ reaches until $(i,-r_i+\varepsilon)$, 
where $r_i=\min\{j: (v_j,v_i)\in E$\}.  (Use $r_0=0$.)
The horizontal segment reaches of $\bb{v_i}$ reaches until $(s_i+\varepsilon,i)$, 
where $s_i=\max\{j: (v_j,v_i)\in E$\}.    (Use $s_k=k$.)  See also Figure~\ref{fig:skeleton}.

\begin{figure}[ht]
\centering
\includegraphics[width=.9\columnwidth,page=1]{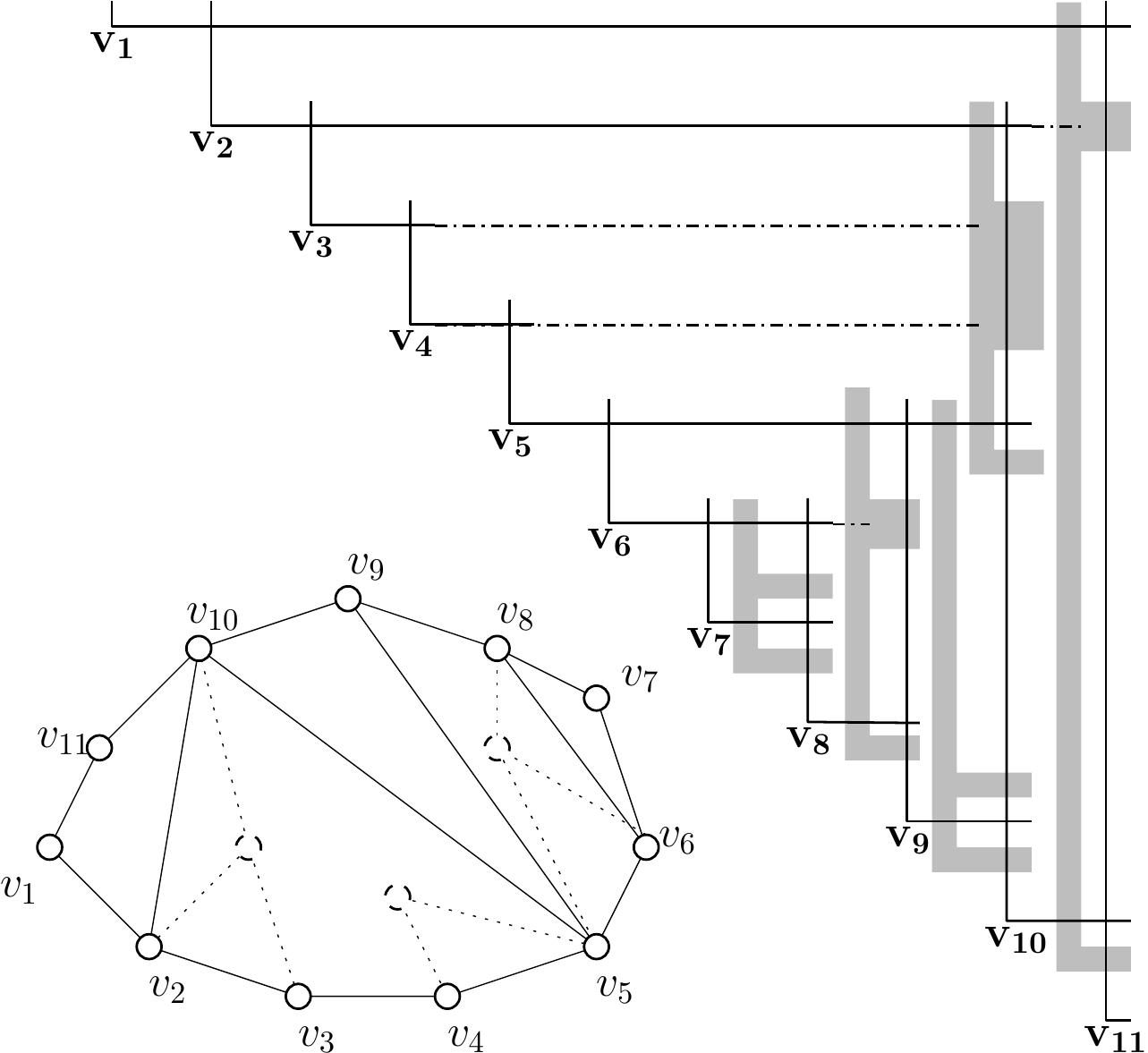}
\caption{Example of an IO-graph and the \{{\tt L}\}-representation of $G[O]$. The IO-private regions are shaded in grey. }
\label{fig:skeleton}
\end{figure}

It is quite easy to see that this is a 1-string
representation.  For every edge $(v_i,v_k)$ with $i<k$
we have created an intersection at $(k,-i)$.  Assume for contradiction that $\bb{v_i}$ and $\bb{v_k}$
intersect for some $(v_i,v_k)\not\in E$ with $i<k$.    Then we must have $s=\max\{j:(v_i,v_j)\in E\}>k$,
else there is no intersection.  Also $r=\min\{j:(v_j,v_k)\in E\}<i$, else there is no intersection.
But then $\{v_i,v_j,v_s,v_r\}$, together with the outer face, form a $K_4$-minor; this is impossible
in an outer planar graph.

Thus we found the \{{\tt L}\}-representation.   To find IO-private regions, we stretch horizontal segments
of curves further as follows.  For vertex $v_i$, set $t_i=\max\{j: v_i$ and $v_j$ are on a common interior face\}.
If $t_i>s_i$, then expand $\bb{v_i}$ horizontally until $t_i-\varepsilon$.  To see that this does not
introduce new crossings, observe that adding $(v_i,v_{t_i})$ to the graph would not destroy outerplanarity, 
since the edge could be routed inside the common face.     The $\{{\tt L}\}$-representation of such an
expanded graph would contain the constructed one and also contain the added segment.  Therefore the added
segment cannot intersect any other curves.

After stretching all curves horizontally in this way, 
an IO-private region for each interior face $f$ can then be inserted to the left of the
vertical segment of $\bb{v_j}$, where $v_j$ is the vertex on $f$ with maximal index;
see also Figure~\ref{fig:skeleton}.
\end{proof}

Now we can prove Theorem~\ref{thm:IO}, i.e., we can show that every IO-graph $G$ has an
\{{\tt L}\}-representation.  Start with the \{{\tt L}\}-representation of $G[O]$
of Lemma~\ref{lem:skeleton-private}.  We add the interior vertices $v_1,\dots,v_{n-k}$ to this
in arbitrary order, maintaining the following invariant:
\begin{quotation}
For every interior face of the current graph
there exists an IO-private region.
\end{quotation}
Clearly this invariant holds for the representation of $G[O]$.
Let $v$ be the next interior vertex to be added, and let $f$ be the face where it should be inserted.
By induction there exists a IO-private region $P_0$ for face $f$ such that the curves $\bb{x_1},\dots,\bb{x_d}$
that intersect $P_0$ include the curves of all exterior vertices that are on $f$, in CCW order.
We need to place an \sw{}-curve $\bb{v}$ into $P_0$, intersecting curves of neighbours of $v$ and nothing else,
and then find IO-private regions for every newly created face. 

Since the interior vertices form an independent set, all neighbours of $v$ are on the outer face,
and hence belong to $\{x_1,\dots,x_d\}$.  Since $G$ is 2-connected, $v$ has at least two such neighbours.
We have two cases.  

\medskip\noindent{\bf Case 1.} If $(v,x_d)$ is not an edge,
then $\bb{v}$ is a vertical segment that extends from the topmost to the bottommost of
the curves of its neighbours, and intersects these curves after expanding them rightwards.

Since the order of $\bb{x_1},\dots,\bb{x_d}$ is CCW
around the outer face, for every newly created face $f'$ incident to $v$ we have a region inside $P_0$
in which the curves of outer face vertices on $f'$ appear in CCW order. IO-private regions for these
faces can be found as shown in Figure~\ref{fig:inserting}(top).  Note that some of these private
regions intersect $v$ while others do not; both are acceptable since $v$ is on those faces, but not
an exterior vertex.

\begin{figure}[ht]
\centering
\includegraphics[width=0.8\columnwidth,page=2]{IO.pdf}
\includegraphics[width=0.8\columnwidth,page=1]{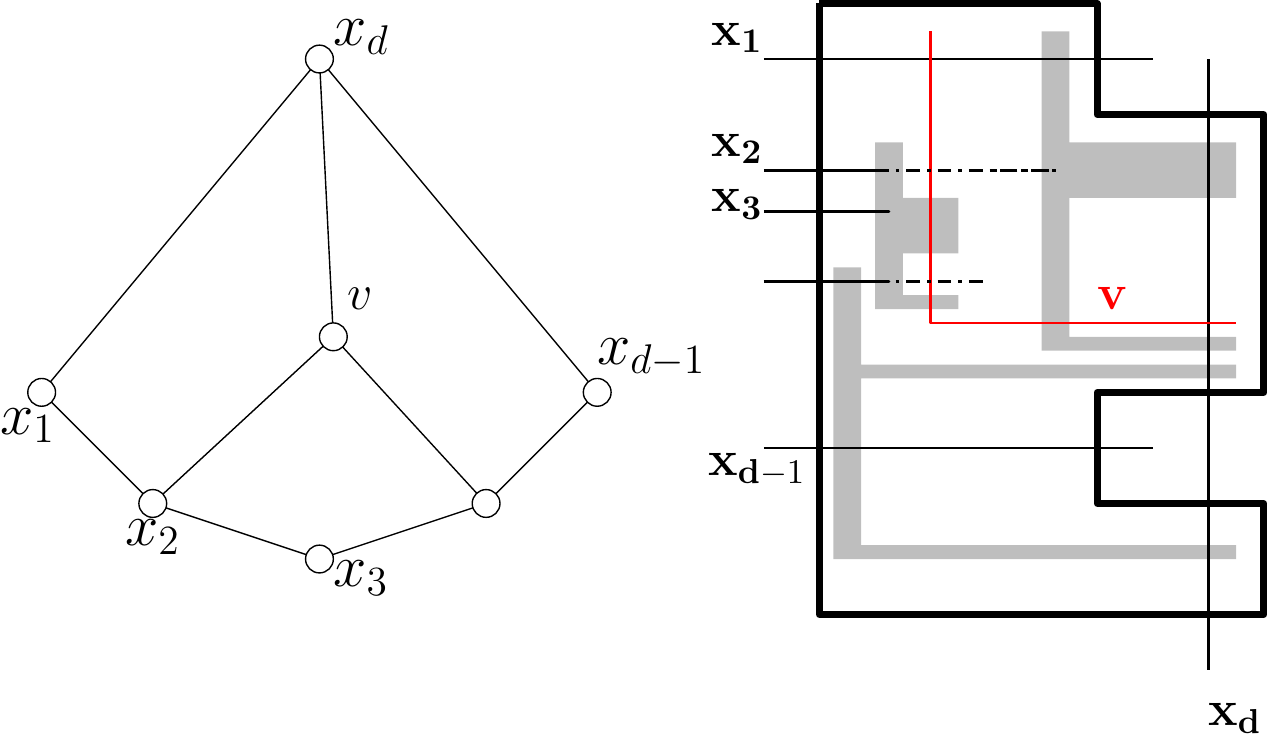}
\includegraphics[width=0.8\columnwidth,page=3]{IO.pdf}
\caption{Inserting a vertex into a face of an IO-graph. (Top) $v$ is not adjacent to $x_d$.
(Middle) $v$ is adjacent to $x_d$, but not $x_{d-1}$.  (Bottom) $v$ is adjacent to both
$x_d$ and $x_{d-1}$.} 
\label{fig:inserting}
\end{figure}

\medskip\noindent{\bf Case 2.} 
If $(v,x_d)$ is an edge, then $\bb{v}$ is an \sw{}, with the bend below $\bb{x_{d-1}}$ if $(v,x_{d-1})$ is
an edge and above $\bb{x_{d-1}}$ otherwise.  
The vertical segment of $\bb{v}$ extends from this bend to the topmost of
$v$'s neighbours in $\{\bb{x_1},\dots,\bb{x_{d-1}}\}$, and intersects the curves of 
these neighbours after expanding them rightwards.  The horizontal segment extends as
to intersect $\bb{x_d}$.

IO-private regions can again be found easily, see
Figure~\ref{fig:inserting}(middle and bottom).

\medskip

Repeating this insertion operation for all interior vertices hence gives the desired representation of $G$.
\hspace*{\fill}$\square$

\section{Halin graphs}
\label{sec:Halin}

A {\em Halin-graph} \cite{cit:halin} is a graph obtained by taking a tree $T$ with $n\geq 3$ vertices
that has no vertex of degree 2 and connecting the leaves in a cycle.    Such graphs were originally
of interest since they are minimally 3-connected, but it was later shown that they are also
planar partial 3-trees \cite{cit:bod88}.  

We now prove Theorem~\ref{thm:Halin} and show that any Halin-graph $G$ has a \{\nec,\sw\}-representation.
We note here that our construction works even if $T$ has some vertices of degree 2.
Fix an embedding of $G$ such that the outer face is the cycle $C$ connecting the leaves of tree $T$.
Enumerate the outer face as $v_1,\dots,v_k$ in CCW order. Since every exterior vertex was a leaf of $T$, vertex $v_k$ has degree 3;
let $r$ be the interior vertex that is a neighbours of $v_k$.  Root $T$ at $r$ and enumerate the 
vertices of $T$ in post-order as $w_1,\dots,w_{n}$, starting
with the leaves (which are $v_1,\dots,v_k$) and ending with $r$.  

Let $G_i$ be the graph induced by $w_1,\dots,w_i$.  Call vertex $v_j$ {\em unfinished} in $G_i$
if it has a neighbour in $G-G_i$.  For $i=k,\dots,n$, we create an \{{\tt L}\}-representation of $G_i-(v_1,v_k)$
that satisfies the following:  
\begin{quotation}
For any unfinished vertex $v$, curve $\bb{v}$ ends in a horizontal
ray, and the top-to-bottom order of these rays corresponds to the CW order of the unfinished
vertices on the outer face while walking from $v_1$ to $v_{k}$.
\end{quotation}
The \{{\tt L}\}-representation of $G_{k}-(v_1,v_k)$ (i.e., the path $v_1,\dots,v_{k}$) is obtained easily
by placing the bend for $\bb{v_i}$ at $(i,-i)$, giving the vertical segment length $1+\varepsilon$
and leaving the horizontal segment as a ray as desired.  To add vertex $w_i$ for $i>k$,
let $x_1,\dots,x_d$ be its children in $T$; their curves have been placed already.  
Insert a vertical segment for $\bb{w_i}$ with $x$-coordinate $i$,
and extending from just
below the lowest curve of $\bb{x_1},\dots,\bb{x_d}$ to just above the highest.    The rays of
$\bb{x_1}, \dots,\bb{x_d}$ end at $x$-coordinate $i+\varepsilon$, while $\bb{w_i}$ appends a horizontal
ray at its lower endpoint.  

Since adding $w_i$ means that $x_1,\dots,x_d$ are now finished (no vertex has two parents), the
invariant holds. Continuing until $i=n$ yields an \{{\tt L}\}-representation of $G-(v_1,v_k)$.  
It remains to add an intersection for edge $(v_1,v_k)$.  To do so, we change the shape of $\bb{v_1}$.
Observe that its vertical segment was not used for any intersection, and that its horizontal segment
can be expanded until $(n+1,-1)$ without intersecting anything except its neighbours.  After this
expansion, we add a vertical segment going downward at its right end.  Since $v_{k}$ is a neighbour of $r$, curve
$\bb{v_k}$ ended when $\bb{r}$ was added, i.e., at $x$-coordinate $n+\varepsilon$, and we can extend it
until $x$-coordinate $n+1+\varepsilon$.    Hence $\bb{v_1}$  and $\bb{v_k}$ can meet at $(n+1,-k)$ if
we change the shape of $\bb{v_1}$ to $\nec{}$.
We have hence proved Theorem~\ref{thm:Halin}.
\hspace*{\fill}$\square$

\begin{figure}[ht]
\centering
\includegraphics[page=1,width=\columnwidth]{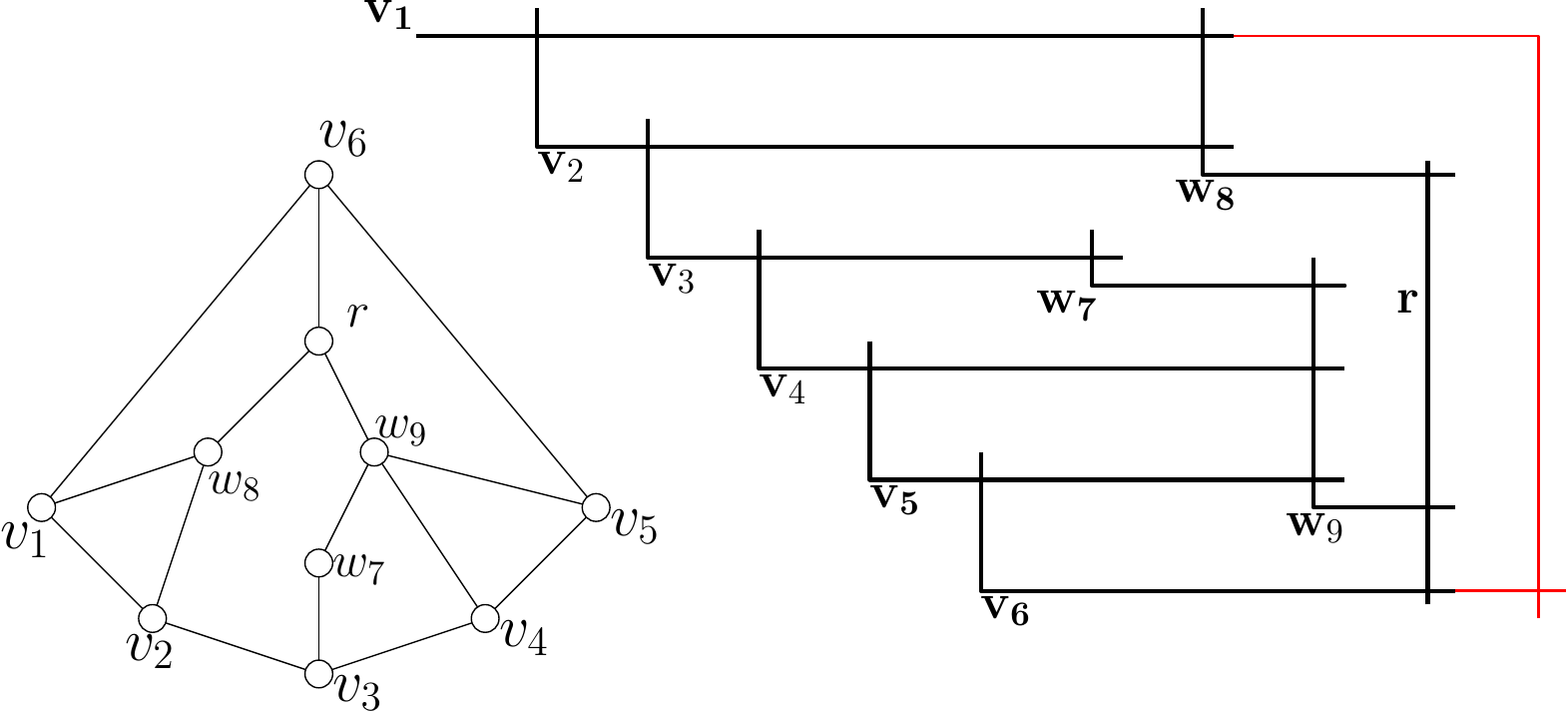}
\caption{Example of an extended Halin-graph and its $\{\sw{},\protect\nec{}\}$-representation,
obtained by changing the curve of $\bb{v_1}$ so that it intersects $\bb{v_k}$. } 
\label{fig:Halin}
\end{figure}

Notice that in the construction for Halin-graphs, any intersection of curves occurs near the
end one of the two curves.  Our result therefore holds not only for Halin graphs, but also for any subgraph of a Halin graph.  

We strongly suspect that Halin-graphs are in fact $\{\tt L\}$-intersection graphs.  We are able
to prove this in the case where $r$ has no neighbours on the outer face other than $v_k$.
We can then change the vertex ordering so that it ends with the children of $r$ in CCW order, followed by
$r$ and $v_k$.  Then $\bb{r}$ can be a horizontal segment crossing the vertical segments
of the children, and $\bb{v_k}$ can be placed entirely different to intersect $\bb{v_{k-1}},
\bb{v_1}$ and $\bb{r}$.  Figure~\ref{fig:Halin2} illustrates this idea.  Generalizing this
to all Halin graphs remains an open problem.

\begin{figure}[ht]
\centering
\includegraphics[page=1,width=.6\columnwidth]{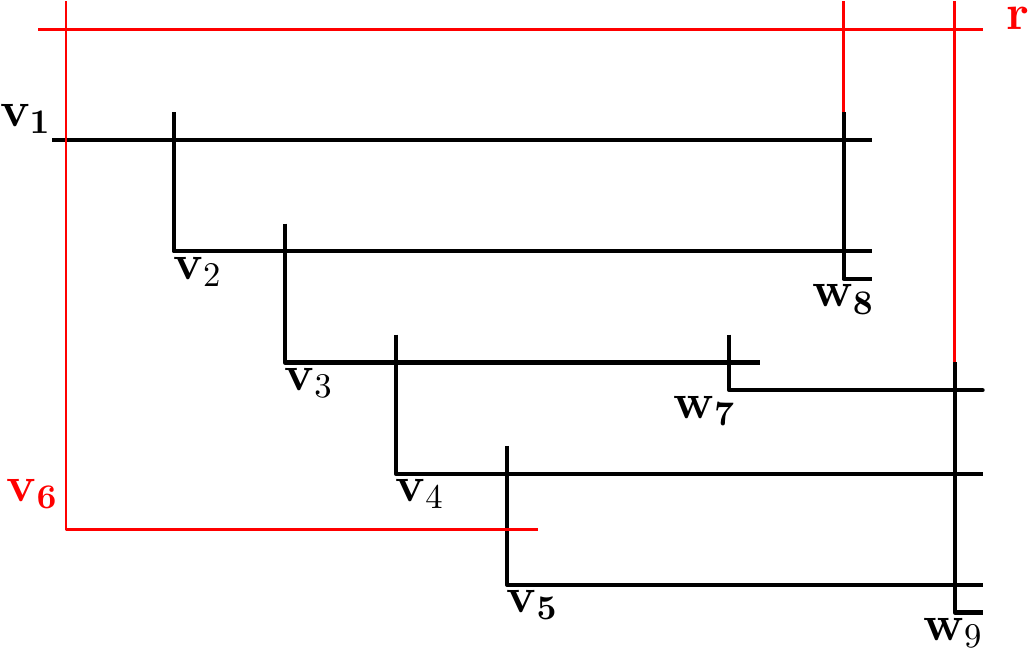}
\caption{$\{\sw\}$-representation,
obtained by changing the curve of $\bb{r}$ and $\bb{v_k}$, if $r$ has no 
other neighbours on the outer face. } 
\label{fig:Halin2}
\end{figure}

\section{Conclusion}
\label{sec:conclusion}

In this paper, we studied 1-string VPG-representations of planar graphs such that curves have at
most one bend.  It is not known whether all planar graphs have such a representation, but
curiously, also no planar graph is known that does not have an $\{\sw{}\}$-representation.
Felsner et al.~\cite{cit:mfcs} asked whether every planar graph has a $\{\nw{},\sw{}\}$-representation
since (as they point out) a positive answer would provide a different proof of Scheinerman's conjecture. 
They proved this for planar 3-trees.  

In this paper, we made another step towards their question and
showed that every planar partial $3$-tree has a $1$-string $B_1$-VPG representation.   We also showed
that IO-graphs and Halin-graphs have $\{\sw{}\}$-representations, except that for Halin-graphs one
vertex curve might be a $\nec{}$.

The obvious direction for future work is to show that all planar partial $3$-trees have 
$\{\sw{}\}$-representations, or at least 
$\{\sw{},\nw{}\}$-representations.
As a first step, an interesting subclass would be those 2-connected planar graphs $G$ where deleting the vertices 
on the outer face leaves a forest; these encompass both IO-graphs and Halin graphs.

Note that all representations constructed in this paper are {\em ordered}, in the sense that the order of intersections along the curves of vertices corresponds to the order of edges around the vertex in a planar embedding.  This is not the case for the 1-string $B_2$-VPG-representations in our earlier construction \cite{cit:socg}.  
One possible avenue towards showing that planar graphs do not always have 
an $\{\sw{}\}$-representation is to restrict the attention to ordered representations first.  Thus,
is there a planar graph that has no 
ordered $\{\sw{}\}$-representation?

\bibliography{cccg-yfr}
\bibliographystyle{plain}

\end{document}